\documentclass[11pt,onecolumn, letter]{article}
\usepackage{psfrag}
\usepackage{graphicx}
\usepackage{color}
\usepackage{url}

\usepackage{amsfonts}
\usepackage{amsthm}
\usepackage{amsmath}
\usepackage{amssymb}
\usepackage{cite}
\graphicspath{{figs/}}
\usepackage{pseudocode}

\newtheorem{theorem}{Theorem}
\newtheorem{lemma}{Lemma}
\newtheorem{remark}{Remark}
\newtheorem{definition}{Definition}[section]
\newtheorem{corollary}{Corollary}

\newtheorem{proposition}{Proposition}

\newcommand{\EE}{\mathbb{E}}
\newcommand{\dom}{{\bf dom}\;}

\newcommand{\dEPI}{\delta_{\mathsf{EPI}, t}}

\newcommand{\Var}{\operatorname{Var}}

\newcommand{\R}{\mathbb{R}}

\title{Wasserstein Stability of the Entropy Power Inequality \\
for Log-Concave Densities}

\author{Thomas~A.~Courtade$^*$, Max Fathi$^{\dagger}$ and Ashwin Pananjady$^*$\\
~\\
$^{*}$Department of Electrical Engineering and Computer Sciences\\
$^{\dagger}$Department of Mathematics\\
University of California, Berkeley
}

\usepackage[]{fullpage}
\begin{document}

\maketitle

\abstract{We establish quantitative stability results for the entropy power inequality (EPI).  Specifically, we show that if uniformly log-concave densities  nearly saturate the EPI, then they must  be close to Gaussian densities in the quadratic Wasserstein distance.  Further, if one of the densities is  log-concave and the other is Gaussian, then the deficit in the EPI can be controlled in terms of the $L^1$-Wasserstein distance.  As a counterpoint, an example shows that the EPI can be unstable with respect to the quadratic Wasserstein distance when densities are uniformly log-concave on sets of measure arbitrarily close to one. Our stability results can be extended to non-log-concave densities, provided certain regularity conditions are met.   The proofs are based on  optimal  transportation.
}

\section{Introduction}

Let $X$ and $Y$ be independent random vectors on $\mathbb{R}^n$ with corresponding probability measures $\mu$ and $\nu$, each absolutely continuous with respect to Lebesgue measure. %
The celebrated entropy power inequality (EPI) proposed by Shannon \cite{shannon48} and proved  by Stam \cite{stam1959some}  asserts that 
\begin{align}
N(\mu*\nu)\geq N(\mu) + N(\nu),\label{ShannonEPI}
\end{align}
where  $N(\mu) := \tfrac{1}{2\pi e} e^{2 h(\mu)/n}$ denotes the  entropy power of $\mu$, and $h(\mu) = h(X) = -\int f \log f$ is entropy associated to the density $f$ of $X\sim \mu$.  %
For a parameter $t\in(0,1)$, let us define
\begin{align} 
\dEPI(\mu, \nu) := h(\sqrt{t} X + \sqrt{1-t}Y) - \Big(t h(X) + (1-t)h(Y)\Big).
\end{align}
Unaware of the works by Shannon, Stam and Blachman \cite{blachman1965convolution}, Lieb \cite{lieb1978proof} rediscovered the EPI by establishing $\dEPI(\mu,\nu)\geq 0$  and noting its equivalence to  \eqref{ShannonEPI}.    Due to the equivalence of the Shannon-Stam and Lieb inequalities, we shall generally refer to both as the EPI.  

 It is  well known that  $\dEPI(\mu,\nu)$ vanishes if and only if $\mu,\nu$ are Gaussian measures that are identical up to translation\footnote{Lieb did not  settle the cases of equality; this was  done later by Carlen and Soffer \cite{carlen1991entropy}.}.    However, despite the  fundamental role the EPI plays in information theory,  few stability estimates are known.  Specifically,  {\em if $\dEPI(\mu, \nu)$ is small, must $\mu$ and $\nu$ be  `close' to Gaussian measures,  which are themselves `close' to each other,   in  a precise and quantitative sense?}  This is our motivating question. %

Toward answering this question, our main result is a  dimension-free, quantitative stability estimate for the EPI.  More specifically, we show that if measures $\mu,\nu$ have uniformly log-concave densities and nearly saturate either form of the EPI, then they must also be close to Gaussian measures in quadratic Wasserstein distance.  We also show that the EPI is \emph{not} stable (with respect to the same criterion) in situations where the densities  \emph{nearly} satisfy the same regularity conditions. A  weaker deficit estimate is obtained involving the $L^1$-Wasserstein distance for log-concave measures when one of the two variables is Gaussian, and dimension-dependent quadratic Wasserstein estimates in certain more general situations. 

Before stating the main results, let us first introduce some notation.  We let 
$\Gamma \equiv \Gamma(\mathbb{R}^n)$ denote the set of centered Gaussian probability measures on $\mathbb{R}^n$, and let $\gamma$ denote the standard Gaussian measure on $\mathbb{R}^n$. That is\footnote{Explicit dependence of quantities on the ambient dimension $n$ will be suppressed in situations where our arguments are the same in all dimensions.}, 
$$d\gamma(x) = d \gamma^n(x) =  e^{-|x|^2/2}\frac{dx}{(2\pi)^{n/2}}.$$  
 Next, we recall that  the quadratic Wasserstein distance between probability measures $\mu,\nu$ is defined according to 
\begin{align*}
W_2(\mu,\nu) = \inf \left( \EE | X-Y  |^2\right)^{1/2}, 
\end{align*}
where $| \cdot |$ denotes the $L^2$ metric on $\R^n$  and the infimum is over all couplings on $X,Y$ with marginal laws $X\sim \mu$ and $Y\sim \nu$.  If $X\sim \mu$ is a centered random vector, then we write $\Sigma_{\mu} = \EE X X^{\top}$ to denote the covariance matrix of $X$.  For two centered probability measures $\mu, \nu$, we define the quantity
\begin{align}
\mathsf{d}^2_F(\mu,\nu) : = \inf_{\theta\in(0,1)}\left\| \sqrt{\theta}\Sigma_{\mu}^{1/2} -  \sqrt{1-\theta}\Sigma_{\nu}^{1/2}\right\|_F^2,\notag
\end{align}
where $\|\cdot\|_F$ denotes  Frobenius norm, to provide a convenient measure of distance between the second order statistics of $\mu,\nu$.  In particular, $\mathsf{d}^2_F(\mu,\nu) = 0$ if and only if  $\Sigma_{\mu}$ and $\Sigma_{\nu}$ are proportional. We  remark here that both forms of the EPI are invariant to translation of the measures $\mu,\nu$.  Thus, our persistent assumption of centered probability measures is for convenience and comes without loss of generality. 

 \subsection*{Organization}  The rest of this paper is organized as follows:  Sections \ref{subs:Main} and \ref{subs:prior} describe our main stability results for log-concave densities and the relationship to previous work, respectively. Section \ref{subs:unstable} gives an example where the EPI is not stable with respect to quadratic Wasserstein distance when  regularity conditions are not met. Section \ref{sec:proofs} gives proofs of our main results and a brief discussion of techniques, which are based on optimal mass transportation.  We conclude in Section \ref{sec:extensions} with extensions of our results to more general settings.

\section{Main Results}
This section describes our main results, and also compares to previously known stability estimates.  Proofs are generally deferred until Section \ref{sec:proofs}. 

\subsection{Stability of the EPI for  log-concave densities} \label{subs:Main}

Our main result is the following:
\begin{theorem}  \label{thm:LogConcaveStable}
Let $\mu = e^{-\varphi}\gamma$ and $\nu = e^{-\psi}\gamma$ be centered probability measures, where $\varphi$ and $\psi$ are convex. 
Then 
\begin{align}
\dEPI(\mu, \nu)  \geq   \frac{t (1 - t) }{2  } \inf_{\gamma_1,\gamma_2\in \Gamma} \Big( W_2^2(\mu,\gamma_1) + W_2^2(\nu,\gamma_2) + W_2^2(\gamma_1,\gamma_2) \Big) . \label{thm1Eq}
\end{align}
\end{theorem}

\begin{remark}
Measures of the form $\mu = e^{-\varphi}\gamma$ for  convex $\varphi$ have several names in the literature.  Such names include `strongly log-concave densities',  `log-concave perturbation of Gaussian',  `uniformly convex potential' and `strongly convex potential' (see \cite[pp. 50-51]{saumard2014log}).  This situation also corresponds to the Bakry-\'Emery condition CD($\rho, \infty$) when the space is $\mathbb{R}^n$. %
\end{remark}

Under the assumptions of the theorem, the three terms in the RHS of \eqref{thm1Eq} explicitly give necessary conditions for the deficit $\dEPI(\mu, \nu) $ to be small.  In particular,  $\mu, \nu$ must each be quantitatively close to Gaussian measures,  which are themselves quantitatively close to one another.   Additionally, $W_2^2$ is additive on product measures, so the estimate \eqref{thm1Eq} is dimension-free, which is compatible with the additivity of $\dEPI$ on product measures. 

Theorem \ref{thm:LogConcaveStable} may be readily adapted to the setting of uniformly log-concave densities.  Toward this end, let $\eta>0$ and recall that  $h(\eta^{1/2} X) = h(X) + \frac{1}{2}\log \eta$, so that  $\dEPI(\mu, \nu)$ is invariant under the rescaling $(X,Y) \to  (\eta^{1/2} X, \eta^{1/2} Y)$. Similarly, if $X\sim \mu$ has density $f$ satisfying the Bakry-\'Emery criterion
\begin{align}
-\nabla^2 \log f \geq \eta \mathrm{I}, \label{BakryEmery}
\end{align}   
then a change of variables reveals that the density $f_{\eta}$ associated to  the rescaled random variable ${\eta}^{1/2} X$ satisfies $-\nabla^2 \log {f_\eta} \geq \mathrm{I}$.  In particular, $f_{\eta} dx = e^{-\varphi} d\gamma$ for some convex function $\varphi$.  Thus,  Theorem \ref{thm:LogConcaveStable} is equivalent to the following:
\begin{corollary}\label{cor:LogConcave}
If $\mu$ and $\nu$ are centered probability measures with densities satisfying  \eqref{BakryEmery},   then 
\begin{align*}
\dEPI(\mu, \nu)  \geq   \eta \frac{t (1 - t) }{2  } \inf_{\gamma_1,\gamma_2\in \Gamma} \Big( W_2^2(\mu,\gamma_1) + W_2^2(\nu,\gamma_2) + W_2^2(\gamma_1,\gamma_2) \Big) .
\end{align*}
\end{corollary}

This result will also apply to certain families of non log-concave measures, see Remark \ref{remark_nonconvex}. 

For convenience, let $\mathsf{d}^2_{W_2}(\mu):= \inf_{\gamma_0\in \Gamma}W_2^2(\mu,\gamma_0)$ denote the squared $W_2$-distance from $\mu$ to the set of centered Gaussian measures.  Using the inequality $(a + b + c)^2 \leq 3 (a^2+b^2+c^2)$ and the triangle inequality for $W_2$, we may conclude a weaker, but potentially more convenient variant of Corollary \ref{cor:LogConcave}.
\begin{corollary}\label{cor:LogConcave2}
If $\mu$ and $\nu$ are centered probability measures with densities satisfying  \eqref{BakryEmery},   then 
\begin{align}
\dEPI(\mu, \nu)  \geq   \eta \frac{t (1 - t) }{8  }  \Big(\mathsf{d}^2_{W_2}(\mu) + \mathsf{d}^2_{W_2}(\nu) + W_2^2(\mu,\nu) \Big) .\label{cor2Ineq}
\end{align}
\end{corollary}
  
Shannon's form  of the entropy power inequality \eqref{ShannonEPI}  is oftentimes preferred to  Lieb's inequality for  applications in information theory.   Starting with Corollary \ref{cor:LogConcave2}, we may establish an analogous estimate for Shannon's EPI.    %
\begin{corollary} \label{cor:entropyPower}
Let $\mu$  and $\nu$ be centered probability measures on $\mathbb{R}^n$ satisfying   \eqref{BakryEmery} with parameters $\eta_{\mu}$ and $\eta_{\nu}$, respectively.  Then, 
\begin{align}
N(\mu*\nu) \geq \left(N(\mu)+N(\nu)\right) \Delta_{\mathsf{EPI}}(\mu,\nu),\notag
\end{align}
where 
\begin{align}
\Delta_{\mathsf{EPI}}(\mu,\nu) := \exp\left( \frac{\min\{\theta \eta_{\mu}, (1-\theta) \eta_{\nu}\}   }{4 n} \Big( (1-\theta) \mathsf{d}^2_{W_2}(\mu)+\theta \mathsf{d}^2_{W_2}(\nu)  + \mathsf{d}^2_F(\mu,\nu)  \Big) \right) ,\label{EPIdef}
\end{align}
and $\theta$ is chosen to satisfy $\theta/(1-\theta) = N(\mu)/N(\nu)$.
\end{corollary}
\begin{remark}
Equality is attained in \eqref{ShannonEPI} if and only if $\mu, \nu$ are Gaussian with proportional covariances.  Under the stated assumptions of log-concavity, these  conditions are explicitly captured by the last three  terms in \eqref{EPIdef}.
  \end{remark}
	
We also derive a  stability estimate when one variable is simply log-concave and the other variable is Gaussian, involving the $L^1$-Wasserstein distance: 
\begin{align*}
W_1(\mu,\nu) = \inf \EE [| X-Y  | ], 
\end{align*}
where $| \cdot |$ denotes the $L^2$ metric on $\R^n$  and the infimum is over all couplings on $X,Y$ with marginal laws $X\sim \mu$ and $Y\sim \nu$.  

\begin{theorem} \label{thm_w1}
For any log-concave centered random variable with law $\mu$, we have
$$\dEPI(\mu, \gamma) \geq C \, t(1-t) \min(W^2_1(\mu, \gamma), 1),$$
with $C$ a numerical constant that does not depend on $\mu$. 
\end{theorem}

This estimate is reminiscent of the deficit estimates on Talagrand's inequality of \cite{fathi2014quantitative, cordero2015quantitative}, with a remainder term that stays bounded when the distance becomes large. The advantage of this estimate is that it does not rely on any quantitative information on $\mu$, only on the fact that it is log-concave. 

\subsection{Relation to Prior Work} \label{subs:prior}

As remarked above, a few stability estimates are known for the EPI.  Here, we review those that are most relevant and  comment on the relationship to our results.  To begin,  we mention a  stability result due to Toscani \cite{toscani2015strengthened}, which asserts for probability measures $\mu,\nu$ with log-concave densities, there is a function $R$ such that
\begin{align}
N(\mu*\nu) \geq \left(N(\mu)+N(\nu)\right) R(\mu,\nu),\notag
\end{align}
where $R(\mu,\nu)=1$   only if $\mu,\nu$ are Gaussian measures.   However, the function $R(\mu,\nu)$ is quite complicated\footnote{$R(\mu,\nu)$ is expressed in terms of integrals of nonlinear functionals evaluated along the evolutes of $\mu,\nu$ under the heat semigroup.}, and does not explicitly control the distance  of $\mu,\nu$ to the space of Gaussian measures.  Toscani leaves this as an open problem  \cite[Remark 7]{toscani2015strengthened}.   Corollary  \ref{cor:entropyPower}   provides a satisfactory answer to his problem when  $\mu,\nu$  satisfy the Bakry-\'Emery criterion \eqref{BakryEmery} for some parameter $\eta>0$.   Similarly, Theorem \ref{thm_w1}  provides an answer when  one of the measures is log-concave and the other Gaussian.

Next,  we compare  Corollary \ref{cor:LogConcave}   to the main result of Ball and Nguyen \cite{ball2012entropy}, which states that  if $\mu$ is a centered  isotropic probability measure (i.e., $\Sigma_{\mu} = \mathrm{I}$) with spectral gap $\lambda$ and  log-concave density (not necessarily uniformly), then 
\begin{align}
\delta_{{\sf EPI},1/2}(\mu,\mu) \geq \frac{\lambda}{4(1 + \lambda)} D(\mu\|\gamma) \geq \frac{\lambda}{8(1 + \lambda)} W_2^2(\mu,\gamma), \label{BallResult}
\end{align}
where $D(\mu\|\gamma) = \int d\mu \log\frac{d \mu}{d \gamma}$ is the relative entropy between $\mu$ and $\gamma$, and the second inequality is due to Talagrand.  Now,  if $\mu$  satisfies the Bakry-\'Emery criterion \eqref{BakryEmery},    then  Corollary \ref{cor:LogConcave} yields the  similar bound 
\begin{align}
\delta_{{\sf EPI},1/2}(\mu,\mu) \geq  \frac{\eta }{4} \inf_{\gamma_0\in \Gamma}W_2^2(\mu,\gamma_0).\notag
\end{align}
Given the similarity, Corollary \ref{cor:LogConcave} may be viewed as an extension of \eqref{BallResult} to non-identical measures and all parameters $t \in (0,1)$. However, two points should be mentioned: (i) a stability  estimate with respect to $W_2$  is weaker than one involving relative entropy; and (ii) the Bakry-\'Emery criterion for $\eta>0$ implies a positive spectral gap, but not vice versa.  It is  interesting to ask   whether the hypothesis of Corollary \ref{cor:LogConcave} can be weakened to  require only a  spectral gap;  the result of Ball and Nguyen and a similar earlier result by Ball, Barthe and Naor \cite{ball2003entropy} in dimension one provides  some grounds for cautious optimism. In Section 4, we shall obtain a one-dimensional result for non log-concave measures under a stronger assumption than \cite{ball2003entropy} (namely a Cheeger isoperimetric inequality), but with the advantage of being valid for non-identical measures. 

The two results mentioned above  assume log-concave densities, as do we.  In contrast, the refined EPI established in \cite{courtadeStrengtheningISIT2016}  provides a qualitative stability estimate for the EPI   when $\mu$ is arbitrary and $\nu$ is Gaussian.  However, the deficit is quantified in terms of the so-called \emph{strong data processing function}, and is therefore not directly comparable to the present results.  Nevertheless, a noteworthy consequence is a \emph{reverse} entropy power inequality, which does bear some resemblance to the result of Corollary \ref{cor:entropyPower}.  In particular, for arbitrary probability measures $\mu,\nu$ on $\mathbb{R}^n$ with finite second moments, it was shown in \cite{courtade2016links} that
\begin{align}
N(\mu* \nu)\leq \left(N(\mu) + N(\nu) \right)\left( (1-\theta)\mathsf{p}(\mu) + \theta \mathsf{p}(\nu)  \right),\label{REPI}
\end{align}
where $\theta$ is the same as in the definition of $\Delta_{\mathsf{EPI}}(\mu,\nu)$ and  $\mathsf{p}(\mu) := \frac{1}{n}N(\mu)J(\mu)\geq 1$ is the \emph{Stam defect}, with $J(\mu)$ denoting Fisher information.   We have $\mathsf{p}(\mu) = 1$ only if $\mu$ is Gaussian, and thus $\mathsf{p}(\mu)$ may  reasonably be interpreted as a measure of how far $\mu$ is from the set of Gaussian measures.  Thus, the deficit term $(1-\theta)\mathsf{p}(\mu) + \theta \mathsf{p}(\nu)$ in \eqref{REPI} bears a pleasant resemblance to the deficit term  $(1-\theta) \mathsf{d}^2_{W_2}(\mu)+\theta \mathsf{d}^2_{W_2}(\nu)$ in Corollary \ref{cor:entropyPower}.  Importantly, though,  the former upper is an bound on $N(\mu*\nu)$, while the latter yields a lower bound. 

We would be remiss to not mention that the inequality $\mathsf{p}(\mu)  \geq 1$ mentioned above  is known as {Stam's inequality}, and is equivalent to Gross' celebrated logarithmic Sobolev inequality for Gaussian measure \cite{gross1975logarithmic, carlen1991superadditivity}.  Taking $\mu = \nu$ in \eqref{REPI} gives the sharpening $\mathsf{p}(\mu)  \geq \exp\left(\frac{2}{n} \dEPI(\mu,\mu)  \right)$, holding for any probability measure $\mu$ with finite second moment. Equivalently,  if $\mu$ is centered, then 
\begin{align}
\frac{1}{2} I(\mu \| \gamma) \geq D(\mu \| \gamma) + \dEPI(\mu, \mu), \notag
\end{align}
where $I(\mu \| \gamma)$ denotes the relative Fisher information of $\mu$ with respect to $\gamma$.  When $\mu$ satisfies the Bakry-\'Emery criterion \eqref{BakryEmery}, we have a dimension-free quantitative stability result for the logarithmic Sobolev inequality $\tfrac{1}{2} I(\mu \| \gamma) \geq D(\mu \| \gamma)$.  This is an improvement upon the main result of Indrei and Marcon \cite{indrei2013quantitative}, who consider the subset of densities satisfying \eqref{BakryEmery} for some parameter $\eta>0$, whose Hessians are also uniformly upper bounded.    Unfortunately, this improvement is already obsolete, as Fathi, Indrei and Ledoux \cite{fathi2014quantitative} have recently shown that a similar result holds for all probability measures with positive spectral gap.  Interestingly though, \eqref{REPI} does imply a general upper bound on $\dEPI(\mu, \nu)$ involving Fisher informations.  Specifically, for arbitrary probability measures $\mu,\nu$ with finite second moments,
\begin{align}
\dEPI(\mu, \nu) \leq (1-t)\left( \frac{1}{2} I(\mu \| \gamma) - D(\mu \| \gamma) \right) + t\left( \frac{1}{2} I(\nu \| \gamma) - D(\nu \| \gamma) \right). \notag
\end{align}
See \cite{courtade2016links} for details. Other deficit estimates for the logarithmic Sobolev inequality have been obtained in \cite{bobkov2014lsi, bolley2015lsi}. 

If $X\sim \mu$ is a radially symmetric random vector on $\mathbb{R}^n$, $n\geq 2$, satisfying modest regularity conditions (e.g,. convolution with a Gaussian measure of small variance is sufficient), then it was recently established in \cite{courtade2016entropy} that, for any $\varepsilon>0$
\begin{align}
\delta_{{\sf EPI},1/2}(\mu,\mu) \geq  C_{\varepsilon}(\mu) n^{\varepsilon} D^{1+\varepsilon}(\mu\| \gamma_{\mu}), \label{radial}
\end{align}
where $\gamma_{\mu}$ denotes the Gaussian measure with the same covariance as $\mu$, and $C_{\varepsilon}(\mu)$ is an explicit function that depends only on $\varepsilon$, a finite number of moments of $\mu$, and its regularity.   This closely parallels quantitative estimates on entropy production in the Boltzmann equation \cite{toscani1999sharp, villani2003cercignani}. Neither \eqref{thm1Eq} nor \eqref{radial}   imply the other since the hypotheses required are quite different (strong log-concavity vs. radial symmetry).  However, both results do give quantitative bounds on entropy production under convolution in terms of a distance from  Gaussian measures.  In general, the constants in \eqref{thm1Eq} will be much better than those in \eqref{radial} which, although numerical, can be quite small.  We return to the setting of radially symmetric measures in Section \ref{sec:extensions}.

Finally, we mention Carlen and Soffer's  qualitative stability estimate for the EPI that holds under general conditions \cite{carlen1991entropy}.  Roughly speaking, their result is the following:  if probability measures $\mu,\nu$ on $\mathbb{R}^n$ are isotropic with  Fisher informations upper bounded by $J_0$, then there is a nonlinear function $\Theta : \mathbb{R}\to [0,\infty)$, strictly increasing from $0$, depending only on the dimension $n$, the parameter $t$, the number $J_0$ and smoothness and decay properties of    $\mu,\nu$ that satisfies
\begin{align}
\dEPI(\mu,\nu)\geq \Theta(D(\mu\| \gamma)).\notag
\end{align}
The construction of the function $\Theta$ relies on a compactness argument, and therefore is non-explicit.  As such, it is again not directly comparable to our results.  However, it does settle cases of equality. 

\subsection{Instability of the EPI: An Example} \label{subs:unstable}

As a counterpoint to Theorem \ref{thm:LogConcaveStable} and to provide justification for the regularity assumptions therein, we observe that there are probability measures that satisfy the hypotheses required in Theorem \ref{thm:LogConcaveStable} on sets of measure arbitrarily close to one, but severely violate its conclusion. 
\begin{proposition}
There is a sequence of probability measures  $(\mu_{\epsilon})_{\epsilon>0}$  on $\mathbb{R}$ with finite and uniformly bounded  entropies and second moments such that  
\begin{enumerate}
\item The measures $\mu_{\epsilon}$ essentially satisfy the Bakry-\'Emery criterion \eqref{BakryEmery} for $\eta=1$.  That is, $\lim_{\epsilon\downarrow 0} \mu_{\epsilon}(\Omega_{\epsilon})  = 1$,  where $\Omega_{\epsilon} := \{ x \mid - \frac{d^2}{dx^2} \log f_{\epsilon}(x) \geq 1\}$ with $d \mu_{\epsilon} = f_{\epsilon}d x$.
\item The measures $\mu_{\epsilon}$  saturate the EPI as $\epsilon$ approaches zero.  That is,  $\lim_{\epsilon\downarrow 0}  \dEPI(\mu_{\epsilon}, \mu_{\epsilon}) = 0$ for all $t\in(0,1)$.
\item The measures $\mu_{\epsilon}$  are bounded away from Gaussians in the $W_2$ metric; specifically,  $\liminf_{\epsilon\downarrow 0}   \inf_{\gamma_0 \in \Gamma } W^2_2(\mu_{\epsilon} , \gamma_{0})> 1/3$. 
 \end{enumerate}
\end{proposition}
 
We remark that the measures $(\mu_{\epsilon})_{\epsilon>0}$ in the proposition are not necessarily pathological.  In fact, it suffices to consider simple Gaussian mixtures that approximate a Gaussian measure, albeit with heavy tails. 
 
 \begin{proof}
Define the density $f_{\epsilon}$ as
\begin{align}
f_{\epsilon}(x) = \epsilon \frac{\sqrt{\epsilon}}{\sqrt{\pi}}e^{-\epsilon x^2 } + (1-\epsilon) \frac{\sqrt{1-\epsilon}}{\sqrt{\pi}}e^{-(1-\epsilon) x^2 }.
\end{align}
Evidently, $f_{\epsilon}$ is a Gaussian mixture having unit variance; the mixture components have variance $(2\epsilon)^{-1}$ and $(2(1-\epsilon))^{-1}$, respectively.   

\emph{Proof of 1.} On any interval, as $\epsilon\downarrow 0$, the densities $(f_{\epsilon})_{\epsilon>0}$ and their derivatives converge uniformly to those of the Gaussian density with variance $1/2$.  Therefore, 
\begin{align}
-\lim_{\epsilon\downarrow 0} f''_{\epsilon}(x) = 2~~~~~\forall x\in \mathbb{R}.
\end{align}
Since the measures $\mu_{\epsilon}$ converge weakly to a Gaussian measure with variance $1/2$, it is straightforward to conclude that   $\lim_{\epsilon\downarrow 0} \mu_{\epsilon}(\Omega_{\epsilon})  = 1$, where $\Omega_{\epsilon}$ is defined as in the statement of the proposition. 

\emph{Proof of 2.}  
This  follows immediately  from \cite[Theorem 1]{godavarti2004convergence} %
due to pointwise convergence of uniformly bounded densities and uniformly bounded second moments.

\emph{Proof of 3.} Let $m_p(\mu)$ denote the $p^{\mathrm{th}}$ absolute moment associated to $\mu$.  For conjugate exponents $p, q\geq 1$, H\"older's inequality implies
\begin{align}
W_2^2(\mu_{\epsilon},\gamma_s) \geq s + 1 - 2 m^{1/p}_p(\mu_{\epsilon})m^{1/q}_q(\gamma_s),\notag
\end{align}
where $\gamma_s$ is the Gaussian measure with variance $s$. Fix $p = 3/2$ and $q=3$.  By the dominated convergence theorem, we have  $\lim_{\epsilon\downarrow 0} m_{3/2}(\mu_{\epsilon}) =  m_{3/2}(\gamma_{1/2})$.  Thus, using the  characterization of Gaussian moments, we have
\begin{align}
W_2^2(\mu_{\epsilon},\gamma_s) \geq s + 1 - 2 \sqrt{\frac{2s}{\pi}} \left( \Gamma(5/4)\right)^{4/3},\notag
\end{align}
where $\Gamma(\cdot)$ denotes the Gamma function.  The minimum is achieved at $s = \frac{2}{\pi} \left( \Gamma(5/4)\right)^{4/3}$, so that 
\begin{align}
 \inf_{s> 0} W_2^2(\mu_{\epsilon},\gamma_s) \geq 1 - \frac{2}{\pi}  \left( \Gamma(5/4)\right)^{4/3} \approx 0.441562 > 1/3.\notag
\end{align}
\end{proof}
\begin{remark}
Our construction of $f_{\epsilon}$ is closely related to the counterexamples proposed by Bobylev and Cercignani in their disproof of Cercignani's conjecture on entropy production in the Boltzmann equation \cite{bobylev1999rate}.  This construction also appeared in the context of the Boltzmann equation in \cite[Proposition 23]{Carlen201085}.
\end{remark}

\section{Discussion and Proofs}\label{sec:proofs}

The remainder of this paper makes use of ideas from optimal transport, and reader familiarity assumed.  The unfamiliar reader is directed to the comprehensive introductions \cite{villani2003topics, villani2008optimal}.     We recall that a map $T:\R^n\to\R^n$ is said to transport a measure $\mu$ to $\nu$ if the pushforward of $\mu$ under $T$ is $\nu$ (i.e., $\nu = T \# \mu$).

Our starting point comes from a recent paper of Rioul \cite{rioul2016yet}.  Through an impressively short sequence of direct but carefully chosen steps, Rioul recently gave a new proof of the EPI based on transportation of measures.  From his proof,  we may readily distill the following:
\begin{lemma}
Let $T_1 : \mathbb{R}^n\to \mathbb{R}^n$ and $T_2 : \mathbb{R}^n\to \mathbb{R}^n$ be diffeomorphisms  satisfying $\mu = T_1 \# \gamma$ and $\nu = T_2 \# \gamma$.  If $\mu$ and $\nu$ have finite entropies, then
\begin{align}
\dEPI(\mu, \nu)  &\geq \EE \log \frac{\det (t \nabla T_1(X^*) +(1-t)  \nabla T_2(Y^*)  )}{ \det (\nabla T_1(X^*))^t  \det (\nabla T_2(Y^*))^{1-t}}, \label{RioulDeficit}
\end{align}
where $X^*\sim \gamma$ and $Y^*\sim \gamma$ are independent.
\end{lemma}
\begin{remark}
For a vector-valued map $\Phi=(\Phi_1, \Phi_2, \dots, \Phi_n) : \mathbb{R}^n\to \mathbb{R}^n$, we write $\nabla \Phi$ to denote its Jacobian.  That is, $(\nabla \Phi (x))_{ij} = \frac{\partial}{\partial x_i} \Phi_j(x)$.
\end{remark}

In words, \eqref{RioulDeficit} shows that the deficit in the EPI can always be bounded from below by a function of the Jacobians $\nabla T_1$ and $\nabla T_2$, where $T_1$ and $T_2$ are are invertible and differentiable maps that transport measures $\gamma$ to $\mu$ and $\gamma$ to $\nu$, respectively. 

When $T_1$ and $T_2$ are Kn\"othe maps (see \cite{knothe1957contributions, rosenblatt1952remarks, villani2008optimal}), the Jacobians $\nabla T_1$ and $\nabla T_2$ are upper triangular matrices with positive diagonal entries.  Using this property, Rioul concludes $\dEPI(\mu,\nu)\geq 0$ using concavity of the logarithm applied to the eigenvalues (diagonal entries) of $\nabla T_1$ and $\nabla T_2$. By strict concavity of the logarithm, saturation of this inequality implies  the diagonal entries of $\nabla T_1$ and $\nabla T_2$ must be equal almost everywhere.  Combining this information with the fact that a relative entropy term (omitted above) must vanish, Rioul recovers the well known necessary and sufficient conditions for $\dEPI(\mu, \nu)$ to vanish.  Specifically, $\mu$ and $\nu$ must be Gaussian measures, equal up to translation.  

In our proof, instead of the Kn\"othe map, we shall use the Brenier map from optimal transport theory, which has a useful rigid structure: 

\begin{theorem}[Brenier-McCann, \cite{brenier1991polar, mccann1995existence, villani2003topics}]
Consider two probability measures $\mu, \nu$ on $\mathbb{R}^n$, and assume that $\mu$ is absolutely continuous with respect to the Lebesgue measure. There exists a unique map $T$ (which we shall call the {\bf Brenier map}) transporting $\mu$ onto $\nu$ that arises as the gradient of a convex lower semicontinuous function. Moreover, this map is such that
$$W^2 _2(\mu, \nu) = \mathbb{E}[|X - T(X)|^2],$$
where $X$ is a random variable with law $\mu$, and therefore $T(X)$ has law $\nu$. In other words, $(X, T(X))$ is an optimal coupling for the Wasserstein distance $W_2$. 
\end{theorem}

In contrast to Rioul's argument based on Kn\"othe maps, if $T_1$ and $T_2$ are taken instead to be Brenier maps (again, transporting $\gamma$ to $\mu$ and $\gamma$ to $\nu$, respectively), then the Jacobians $\nabla T_1$ and $\nabla T_2$ are symmetric positive definite by the Brenier-McCann Theorem.  Thus, concavity of the log-determinant function on the positive semidefinite cone immediately gives the EPI from \eqref{RioulDeficit}.  Moreover, by strict concavity of the log-determinant function, equality in the EPI implies  $\nabla T_1(X^*)  = \nabla T_2(Y^*)$ almost everywhere, and are thus constant.  Hence, $T_1$ and $T_2$ are necessarily affine functions, identical up to translation.  This immediately implies  $\dEPI(\mu,\nu)= 0$ only if $\mu,\nu$ are Gaussian measures with identical covariances. 

Unfortunately, while both arguments easily settle cases of equality in the EPI, neither yield quantitative stability estimates.  However,  we note that the Brenier map is generally better suited for establishing quantitative stability in functional inequalities.  Indeed, it was remarked by Figalli, Maggi and Pratelli  in their comparison to Gromov's proof of the isoperimetric inequality  that  the Brenier map is generally more efficient than the Kn\"othe map in establishing quantitative stability estimates    due to its rigid structure \cite{figalli2010mass}.  We shall fruitfully exploit the properties of the Brenier map in our proof of Theorem \ref{thm:LogConcaveStable}.

\subsection{Proof of Theorem \ref{thm:LogConcaveStable}}

The proof of Theorem \ref{thm:LogConcaveStable} is  short, but makes use of  several foundational results from the theory of optimal transport.  We will need the following lemma; a proof can be found  in the Appendix. 
\begin{lemma}\label{lem:FrobNormIneq}
For positive definite matrices $A, B$ and  $t \in [0,1]$, we have
\begin{align}
\log \det(t A + (1 - t) B) \geq t \log \det(A)+ (1 - t) \log \det (B) + \frac{t (1 - t)}{2 \max\{\lambda^2_{\mathrm{max}}(A), \lambda^2_{\mathrm{max}}(B) \}} \| A - B \|_F^2, \nonumber
\end{align}
where $\lambda_{\mathrm{max}}(\cdot)$ denotes the largest eigenvalue. 
\end{lemma}

In addition, we remind the reader that a random vector $X$ having log-concave density enjoys (i) finite second moment (in fact, finite moments of all orders); and (ii) finite entropy $h(X)$. Since Theorem \ref{thm:LogConcaveStable} requires log-concave densities, these conditions will be implicitly assumed throughout the proof. %

\begin{proof}[Proof of Theorem \ref{thm:LogConcaveStable}]  Assume first that the densities $e^{-\varphi}$ and $e^{-\psi}$ are smooth and strictly positive on $\mathbb{R}^n$.  Also, let $X^*\sim \gamma$ and $Y^*\sim \gamma$ be independent.    Define $T_1$ to be the Brenier map transporting $\gamma$ to $\mu$, and let $T_2$ denote the Brenier map transporting $\gamma$ to $\nu$.   We recall here that a Brenier map is always the gradient of a convex function by the Brenier-McCann theorem, and therefore $\nabla T_2$ and $\nabla T_1$ are positive semidefinite since they coincide with Hessians of convex functions.  In fact, since all densities involved are non-vanishing,  they are positive definite. Moreover, when the densities are strictly positive on the whole space, we know by results of Caffarelli \cite{caffarelli1990localization, caffarelli1992regularity} that the maps $T_1$ and $T_2$ are $C^1$-smooth.

Using the assumed smoothness  and convexity of the potentials $\varphi$ and $\psi$, Caffarelli's contraction theorem (see \cite{caffarelli2000monotonicity} and, e.g., \cite[Theorem 9.14]{villani2003topics}) implies $T_1$ and $T_2$ are 1-Lipschitz, so that $\lambda_{\mathrm{max}}(\nabla T_1) \leq 1$ and $\lambda_{\mathrm{max}}(\nabla T_2) \leq 1$.     Therefore,  since $\nabla T_2$ and $\nabla T_1$ are positive definite,  Lemma \ref{lem:FrobNormIneq} yields the following (pointwise) estimate 
\begin{align}
\log \det(t \nabla T_1 + (1 - t) \nabla T_2) \geq t \log \det(\nabla T_1)+ (1 - t) \log \det (\nabla T_2) +  \frac{t (1 - t) }{2  } \| \nabla T_1 - \nabla T_2 \|_F^2. \nonumber
\end{align}
Combined with  \eqref{RioulDeficit} we obtain:
\begin{align*}
\dEPI(\mu, \nu)   \geq  \frac{t (1 - t) }{2  } \EE  \| \nabla (T_1(X^*) - X^*) - \nabla (T_2(Y^*)-Y^*) \|_F^2.
\end{align*}
Now, define matrices $A = \EE [\nabla (T_1(X^*) - X^*)]$ and $B = \EE [\nabla (T_2(Y^*) - Y^*)]$.  By orthogonality, we have
\begin{align*}
&\EE  \| \nabla (T_1(X^*) - X^*) - \nabla (T_2(Y^*)-Y^*) \|_F^2\\
& = \EE  \| \nabla (T_1(X^*) - (I+A)X^*)  -  \nabla (T_2(Y^*)-(I+B)Y^*) \|_F^2 + \| A-B\|_F^2\\
&=\EE  \| \nabla (T_1(X^*) - (I+A)X^*)\|_F^2 + \EE  \|  \nabla (T_2(Y^*)-(I+B)Y^*) \|_F^2 + \| A-B\|_F^2\\
&\geq \EE  |  T_1(X^*) - (I+A)X^*|^2 + \EE  | T_2(Y^*)-(I+B)Y^* |^2 + \| (I+A)-(I+B) \|_F^2.
\end{align*}
The final inequality is due to the  Gaussian  Poincar\'e inequality  $\int |f|^2 d \gamma \leq \int |\nabla f|^2 d \gamma$, holding for every $C^1$-smooth $f: \R^n\to \R$ with mean zero.  Indeed, its application   is justified by  $C^1$-smoothness of the Brenier maps  among log-concave distributions,  
 and the identity $$\EE[ T_1(X^*) - (I+A)X^*] = \int x d\mu(x) - (I+A)\int x d \gamma(x) =0,$$ which holds similarly for $Y^*$.  The desired inequality now follows from the definition of $W_2$ and the  identity $W_2^2(\gamma_{1}, \gamma_{2}) = \| \Sigma_{\gamma_1}^{1/2} - \Sigma_{\gamma_2}^{1/2}  \|_F^2$ (e.g., \cite{dowson1982frechet}).  The application of this identity is valid   because  $I+A$ and $I+B$ are positive definite, a consequence of the positive definiteness of  $\nabla T_1$ and $\nabla T_2$ noted previously.   

Thus, Theorem \ref{thm:LogConcaveStable} holds when the densities are smooth and positive everywhere.  If this is not the case, we may first convolve  $\mu, \nu$ with a Gaussian measure so that the resulting densities will be both smooth and positive. The general result then follows by considering arbitrarily small perturbations, which is justified by Proposition \ref{smallPerturb}, found below. 
\end{proof}

\begin{proposition}\label{smallPerturb}
For a probability measure $\mu$ and $s>0$, let $\mu_s$ denote the probability measure corresponding to  $X+ \sqrt{s} Z$, where $X\sim \mu$ and $Z\sim \gamma$ are independent.  If probability measures $\mu, \nu$ have finite second moments, and 
\begin{align}
 \Delta(\mu, \nu):=\inf_{\gamma_1,\gamma_2\in \Gamma} \Big( W_2^2(\mu,\gamma_1) + W_2^2(\nu,\gamma_2) + W_2^2(\gamma_1,\gamma_2) \Big),\notag
\end{align}
then  $\lim_{s\downarrow 0} \Delta(\mu_s, \nu_s)  =  \Delta(\mu, \nu)$.  Moreover, if $\mu, \nu$ have densities, then $\lim_{s \downarrow 0} \dEPI(\mu_s, \nu_s)  =  \dEPI(\mu, \nu)$.  Finally, if $\mu$ satisfies the Bakry-\'Emery criterion \eqref{BakryEmery} with parameter $\eta >0$, then $\mu_s$ satisfies the Bakry-\'Emery criterion \eqref{BakryEmery} with parameter $\eta/(1+s\eta)$.
\end{proposition}
\begin{proof}
Evidently, if $\gamma_s$ denotes the Gaussian measure with covariance matrix $s \mathrm{I}$, then $\mu_s = \mu *\gamma_s$.    Observe that for any fixed Gaussian measures $\gamma_1,\gamma_0$ (not to be confused with $\gamma_s$), 
\begin{align*}
&W_2^2(\mu,\gamma_1) + W_2^2(\nu,\gamma_2) + W_2^2(\gamma_1,\gamma_2)\\
&\geq W_2^2(\mu* \gamma_s,\gamma_1* \gamma_s) + W_2^2(\nu* \gamma_s,\gamma_2* \gamma_s) + W_2^2(\gamma_1* \gamma_s,\gamma_2* \gamma_s)\\
&\geq \Delta(\mu_s, \nu_s),
\end{align*}
so that $\Delta(\mu, \nu)\geq \Delta(\mu_s, \nu_s)$.  It remains to prove an inequality in the reverse direction.

In general, $W_2^2(\mu, \nu) \geq (\sqrt{\EE |X|^2} - \sqrt{\EE |Y|^2} )^2$ when $X\sim \mu$ and $Y\sim \nu$ by the Cauchy-Schwarz inequality.  Therefore, there is $K = K(\mu,\nu)<\infty$ depending only on the second moments of $\mu$ and $\nu$ such that 
\begin{align*}
 \Delta(\mu, \nu)=\inf_{\gamma_1,\gamma_2\in \Gamma_K} \Big( W_2^2(\mu,\gamma_1) + W_2^2(\nu,\gamma_2) + W_2^2(\gamma_1,\gamma_2) \Big),
\end{align*}
where $\Gamma_K\subset \Gamma$ denotes the set of centered Gaussian measures with second moments bounded by $K$. From this, it is straightforward to argue that for $s$ sufficiently small, 
\begin{align}
 \Delta(\mu_s, \nu_s) &\geq 
\inf_{\gamma_1,\gamma_2\in \Gamma} \Big( W_2^2(\mu,\gamma_1) + W_2^2(\nu,\gamma_2) + W_2^2(\gamma_1,\gamma_2) \Big) - \sqrt{s} C,
\end{align}
where $C = C(\mu,\nu)<\infty$ depends only on the second moments of  $\mu$ and $\nu$.  Thus, the first part of the claim follows.  

The second claim follows immediately  from the fact that $\lim_{s\downarrow 0}h(\mu_s)=h(\mu)$ for any $\mu$ having density and finite second moment. See, e.g., \cite[Lemma 1.2]{carlen1991entropy}. 

The third claim can be found in \cite[Theorem 3.7(b)]{saumard2014log}. %
\end{proof} 

\begin{remark} \label{remark_nonconvex}
The assumption of log-concavity is mainly used to ensure that the optimal transport map is Lipschitz. This can sometimes still be the case in other situations. For example, the recent work \cite{colombo2016lipschitz} shows that this property holds for certain families of bounded perturbations of the Gaussian measure, including the radially symmetric case. 
\end{remark}

\subsection{Proof of Corollary 3}

\begin{proof}[Proof of Corollary \ref{cor:entropyPower}]
Lieb's derivation of $N(X+Y) \geq N(X)+N(Y)$ from the inequality $h(\sqrt{t} \tilde{X} + \sqrt{1-t}\tilde{Y}) \geq t h(\tilde{X}) + (1-t)h(\tilde{Y})$ proceeds by choosing $t$ to satisfy $t/(1-t) = N(X)/N(Y)$ and identifying $\tilde X = t^{-1/2} X$ and $\tilde Y = (1-t)^{-1/2} Y$.  Now, if $X\sim \mu$ satisfies the Bakry-\'Emery criterion \eqref{BakryEmery} with parameter $\eta_{\mu}$, then the density of $\tilde{X}$ satisfies the Bakry-\'Emery criterion with parameter $t \eta_\mu$.  A similar statement holds for $\tilde Y$.   Thus, both  $\tilde{X}$ and  $\tilde{Y}$ satisfy the Bakry-\'Emery criterion with parameter $\min\{t \eta_{\mu}, (1-t) \eta_{\nu}\}$.

Let $\tilde \mu$ and $\tilde \nu$ denote the laws of $\tilde X$ and $\tilde Y$, respectively.  We note that $t(1-t)\mathsf{d}_{W_2}^2(\tilde \mu) = (1-t)\mathsf{d}_{W_2}^2(\mu)$ and $t(1-t)\mathsf{d}_{W_2}^2(\tilde \nu) = t \mathsf{d}_{W_2}^2(\nu)$ by a simple rescaling.  Also, we have the following lower bound on $W_2^2(\tilde \mu, \tilde \nu)$:
\begin{align}
t(1-t) W_2^2(\tilde \mu, \tilde \nu) \geq \mathsf{d}_{F}^2( \mu,  \nu).\notag
\end{align}
This follows from rescaling, the fact that $W_2$ is non-increasing under rescaled convolution, the central limit theorem, weak continuity of $W_2$, the identity for $W_2$ distance between Gaussian measures mentioned in the proof of Theorem \ref{thm:LogConcaveStable}, and finally the definition of $\mathsf{d}_{F}^2$. 

Upon substituting  our choice of $t$ and the above observations into  \eqref{cor2Ineq}, we find
\begin{align*}
 \frac{2}{n}h(\mu * \nu) &\geq \log(N(\mu)+N(\nu)) +  \log(2\pi e) \\
 &~~~+ \frac{ \min\{t \eta_{\mu}, (1-t) \eta_{\nu}\} }{4 n  }  \Big((1-t)\mathsf{d}_{W_2}^2(\mu) +  t \mathsf{d}_{W_2}^2(\nu) + \mathsf{d}_{F}^2( \mu,  \nu) \Big),
\end{align*}
which completes the proof.
 \end{proof}

\subsection{Proof of Theorem \ref{thm_w1}}
\begin{proof}[Proof of Theorem \ref{thm_w1}]
Assume that $X^*\sim \gamma$, and let $T$ be the Brenier map sending $\gamma$ onto $\mu$. For convenience, let us write $\lambda_i$ for the eigenvalues of $\nabla T(x)$, in increasing order (so that $\lambda_n = \lambda_{\mathrm{max}}(\nabla T(x))$). 
The combination of \eqref{RioulDeficit} and Lemma \ref{lem:FrobNormIneq} yields in this case
\begin{equation}
\dEPI(\mu, \gamma) \geq \frac{t(1-t)}{2}\EE \left[\frac{\|\nabla T(X^*) - \mathrm{I} \,\|_F^2}{1 + \lambda_{\mathrm{max}}(\nabla T(X^*) )^2}\right] = \frac{t(1-t)}{2}\EE \left[\frac{\sum_{i=1}^n (\lambda_i-1)^2}{1 + \lambda_{n}^2}\right]. 
\end{equation}
From the $L^1$ Poincar\'e inequality for the Gaussian measure, we have
$$W_1(\mu, \gamma) \leq 2\EE \left[ \sqrt{\sum (\lambda_i - 1)^2} \right].$$
By the Cauchy-Schwarz inequality, we have
\begin{align*}
\EE\left[ \sqrt{\sum (\lambda_i - 1)^2} \right] &\leq \sqrt{\EE [1 + \lambda_n^2]}\sqrt{\EE \left[\frac{\sum (\lambda_i - 1)^2}{1 + \lambda_n^2} \right]} \\
&\leq \sqrt{\EE [1 + \lambda_n^2]}\sqrt{\frac{2}{t(1-t)} \dEPI(\mu, \gamma) }.
\end{align*}
Hence in this situation, if we have an $L^2$ bound on the largest eigenvalue of $\nabla T$,  we can deduce a $W_1$ estimate on the deficit (in contrast to using a uniform bound as in the proof of Theorem \ref{thm:LogConcaveStable}). 

A result of Kolesnikov asserts that $\EE [\lambda_n^2] \leq \frac{3}{2}\EE[\lambda_n]^2$ (see \cite{kolesnikov2014hessian}, Theorem 6.1 and the discussion at the top of page 1526). Moreover, 
$$\EE[\lambda_n] \leq 1 + \EE[|\lambda_n -1|] \leq 1 + \EE \left[ \sqrt{\sum (\lambda_i - 1)^2} \right].$$
From this estimate we deduce
$$\EE \left[ \sqrt{\sum (\lambda_i - 1)^2} \right] \leq \sqrt{4 + 3\EE \left[ \sqrt{\sum (\lambda_i - 1)^2} \right]^2}\sqrt{\frac{2}{t(1-t)} \dEPI(\mu, \gamma) }.$$
Since $r/\sqrt{1 + r^2} \geq c\min(r, 1)$ and $2\EE \left[ \sqrt{\sum (\lambda_i - 1)^2} \right] \geq W_1(\mu, \gamma)$, we deduce the estimate 
$$\sqrt{2\dEPI(\mu, \gamma)/(t(1-t))} \geq C\min(W_1(\mu, \gamma), 1),$$
and the result follows. 
\end{proof}

\section{Extensions}\label{sec:extensions}

The proof of Theorem \ref{thm:LogConcaveStable} uses the fact that under the assumptions the Brenier optimal maps are Lipschitz to bound the square of the  largest eigenvalue $\lambda_{\mathrm{max}}^2$ in the deficit estimate for the log-concavity of the determinant. A natural question is whether we can use weaker assumptions on the map and still obtain a deficit estimate for the EPI. It turns out that we can get an estimate, provided the largest eigenvalue of $\nabla T(x)$ grows at most linearly at infinity. We shall later see that in dimension 1, as well as for multidimensional radially symmetric measures, this assumption of eigenvalue growth holds as soon as the law of the random variable satisfies a Cheeger isoperimetric inequality, which is a stronger assumption than the spectral gap assumption used for the one-dimensional result of \cite{ball2003entropy}, but equivalent in (non-uniformly) log-concave situations. 

A first case in which we establish a deficit estimate is when one of the two variables is Gaussian: 

\begin{proposition} \label{GaussianLinearEigs}
Let $\mu$ be a centered probability measure on $\mathbb{R}^n$, and let $T$ be the Brenier map sending the standard Gaussian measure $\gamma$ onto $\mu$.  If $T$ is $C^1$ and  satisfies the pointwise bound $\lambda_{\mathrm{max}}(\nabla T(x)) \leq c\sqrt{1 + |x|^2}$ for all $x$, for some $c >1$. Then 
$$\dEPI(\mu,\gamma)  \geq \frac{t(1-t)}{8c^2n}W^2_2(\mu,\gamma).$$
\end{proposition}

This estimate can  be compared to Theorem \ref{thm_w1}. Its advantage is that it involves the stronger $W_2$ distance, and that its assumptions may be satisfied in certain non log-concave situations (as we shall later see for one-dimensional random variables), but has the downside of strongly depending on the dimension, and of requiring more quantitative information on the measure $\mu$, via the constant $c$ in the eigenvalue bound. 

\begin{proof}
Using the assumption, and following the same steps as in the proof of Theorem \ref{thm:LogConcaveStable}, we have
\begin{equation*}
\dEPI(\mu, \gamma) \geq \frac{t(1-t)}{2c^2}\mathbb{E}\left[\frac{\|\nabla T(X^*) - \mathrm{I}\|_{F}^2}{1 +\ |X^*|^2}\right].
\end{equation*}

According to Corollary 5.6 in \cite{bonnefont2016sphere} (see also \cite{bonnefont20161d} for the one-dimensional case), the standard Gaussian measure in dimension $n$ satisfies the weighted Poincar\'e inequality
$$\mathbb{E}\left[\frac{|\nabla f(X^*)|^2}{1 + |X^*|^2}\right] \geq \frac{1}{4n}\Var (f(X^*)).$$
Applying this result and following the same steps as in the proof of Theorem \ref{thm:LogConcaveStable} yields

\begin{equation*}
\dEPI(\mu,\gamma) \geq \frac{t(1-t)}{8c^2n}W^2_2(\mu,\gamma).
\end{equation*}
which concludes the proof. 
\end{proof}

We also prove a lower bound when neither of the two measures are Gaussian, but with an even worse dependence on the dimension: 

\begin{proposition}
Let $T_1$ and $T_2$ be the Brenier maps sending  Gaussians to $X$ and $Y$, respectively.  Assume that $X$ and $Y$ are centered and that the maps $T_i$ are $C^1$ and satisfy the pointwise bound $\lambda_{\mathrm{max}}(\nabla T_i(x)) \leq c\sqrt{1 + |x|^2}$ for all $x$, for some $c >1$. Then there is a universal constant $C> 0$ such that
$$\dEPI(\mu,\nu) \geq \frac{Ct(1-t)}{(cn)^2}\underset{\gamma_1, \gamma_2 \in \Gamma}{\inf}(W^2_2(\mu, \gamma_1) + W^2_2(\nu, \gamma_2) + W^2_2(\gamma_1, \gamma_2)).$$
\end{proposition}

\begin{proof}
Let us define 
\begin{align}
&c_n := \mathbb{E}[(1 + |X^*|^2)^{-1}], &A := c_n^{-1}\mathbb{E}\left[\frac{ \nabla T_1(X^*) - \mathrm{I}  }{1 + |X^*|^2}\right], \mbox{~and~} &B := c_n^{-1}\mathbb{E}\left[\frac{  \nabla T_2(Y^*) - \mathrm{I}  }{1 + |Y^*|^2}\right].\notag
\end{align}
As in the  proof of Theorem \ref{thm:LogConcaveStable}, we have
\begin{equation*}
\dEPI(\mu,\nu) \geq \frac{t(1-t)}{2c^2}\mathbb{E}\left[\frac{\|\nabla T_1(X^*) - \mathrm{I} - (\nabla T_2(Y^*) - \mathrm{I}  ) \|_{F}^2}{(1 + |X^*|^2)(1 + |Y^*|^2)}\right],
\end{equation*}
where we have used the na\"ive bound $\lambda_{\mathrm{max}}^2 \leq c^2(1 + |x|^2)(1 + |y|^2)$, which ultimately leads to the worsening of the dependence on the dimension. We then have
\begin{align*}
&\mathbb{E}\left[\frac{\|\nabla T_1(X^*) - \mathrm{I} - (\nabla T_2(Y^*) - \mathrm{I}  ) \|_{F}^2}{(1 + |X^*|^2)(1 + |Y^*|^2)}\right]\\
&= \mathbb{E}\left[\frac{\|\nabla T_1(X^*) - \mathrm{I} -c_nA\|_{F}^2}{(1 + |X^*|^2)(1 + |Y^*|^2)}\right] + \mathbb{E}\left[\frac{\| \nabla T_2(Y^*) - \mathrm{I} -c_nB\|_{F}^2}{(1 + |X^*|^2)(1 + |Y^*|^2)}\right] \\
&~+ c_n^2\|c_n(A - B)\|_F^2-2\mathbb{E}\left\langle \frac{\nabla T_1(X^*) - \mathrm{I}  -c_nA}{1 + |X^*|^2}, \frac{\nabla T_2(Y^*) - \mathrm{I}  - c_nB}{1 + |Y^*|^2} \right\rangle \\
&~+ 2\mathbb{E}\left[\left\langle \frac{c_n(A- B)}{1 + |Y^*|^2}, \frac{\nabla T_1(X^*) - \mathrm{I}  -c_nA}{1 + |X^*|^2} \right\rangle \right] \\
&~- 2\mathbb{E}\left[\left\langle \frac{c_n(A- B)}{1 + |X^*|^2}, \frac{\nabla T_2(Y^*) - \mathrm{I}  -c_n B}{1 + |Y^*|^2} \right\rangle \right] \\ 
& = c_n\mathbb{E}\left[\frac{\|\nabla T_1(X^*) - (\mathrm{I}  + c_nA)\|_{F}^2}{(1 + |X^*|^2)}\right] + c_n\mathbb{E}\left[\frac{\|\nabla T_2(Y^*) - (\mathrm{I}  + c_nB) \|_{F}^2}{(1 + |Y^*|^2)}\right] + c_n^2\|c_nA - c_nB\|_F^2
\end{align*}
and then the proof continues in the same way as for the previous proposition. The constant $c_n$ above is the expectation of $(1 + |X^*|^2)^{-1}$, which is of order $n^{-1}$. 
\end{proof}

To apply these results, we want to know when does the Brenier map satisfy the eigenvalue bound. We shall prove that for one-dimensional measures and for radially symmetric log concave measures, such an assumption holds when the measure satisfies a certain isoperimetric inequality. 

\begin{proposition}
If the law of $X$ is given by an exponential measure $\mu(dx) = \frac{1}{2}\exp(-|x|)dx$, then the Brenier map $T$ transporting a standard Gaussian random variable onto $X$ satisfies the bound
$$T'(x) \leq c\sqrt{1 + x^2}$$
for some $c > 0$ 
\end{proposition}

\begin{proof}
The optimal map from a measure $\mu$ onto a measure $\nu$ with positive densities in dimension one is given by $x \longrightarrow F_{\nu}^{-1}(F_{\mu}(x))$, where $F_{\mu}(x) = \mu(]-\infty, x])$ is the distribution function associated to $\mu$. For the exponential measure the distribution function can be explicitly computed as $F_{\mathrm{exp}}(x) = 1 - e^{-x}/2$ for $x \geq 0$ and $e^{x}/2$ for $x < 0$. 

Consider $x > 0$. A direct computation shows that
$$T'(x) = \frac{e^{-x^2/2}}{\sqrt{2\pi}}\frac{2}{1 - F_{\gamma}(x)}$$
where $F_{\gamma}$ is the distribution function of the standard Gaussian measure. There exists a constant $c$ such that $1-F_{\gamma}(x) \geq \frac{e^{-x^2/2}}{c\sqrt{1 + x^2}}$, and the bound on $T'$ immediately follows. By symmetry, the same bound applies when $x < 0$. 

To prove the lower bound on $1-F_{\gamma}(x)$, we just use the fact that $$1-F_{\gamma}(x) = \frac{e^{-x^2/2}}{\sqrt{2\pi}x} - \frac{1}{\sqrt{2\pi}}\int_x^{\infty}{\frac{e^{-t^2/2}}{t^2}dt} \geq \frac{e^{-x^2/2}}{\sqrt{2\pi}x} - \frac{1-F_{\gamma}(x)}{x^2},$$ and the existence of a suitable constant easily follows. 
\end{proof}

\begin{definition}
A probability measure $\mu$ is said to satisfy a Cheeger isoperimetric inequality with constant $\lambda > 0$ if for any measureable set $A$ we have
$$\mu^+(\partial A) \geq \lambda \mu(A)(1-\mu(A))$$
where $\mu^+(\partial A) := \limsup_{\epsilon\downarrow 0} \frac{\mu(A+ B_{\epsilon}) - \mu(A)}{\epsilon}$, with $B_{\epsilon}$ the ball with center $0$ and radius $\epsilon$. 
\end{definition}

For log-concave measure, Buser's theorem \cite{buser1982isoperimetry} states that satisfying a Cheeger inequality and a Poincar\'e inequality is equivalent, up to universal constants (or more precisely, the extension of Buser's theorem to weighted spaces \cite{ledoux2004buser}). In general, the Cheeger inequality is stronger than the Poincar\'e inequality. It is equivalent (up to universal constants) to the $L^1$ Poincar\'e inequality $\int{|\nabla f|d\mu} \geq c\int{|f|d\mu}$ for every function with average zero. More generally, for log-concave measures the isoperimetric inequality and the exponential concentration property are equivalent \cite{milman2009}. 

\begin{theorem}
Assume $X$ is a one-dimensional random variable with positive density and median $0$, that satisfies a Cheeger inequality with constant $\alpha$. Then the optimal map transporting the exponential measure onto $X$ is $\alpha^{-1}$-Lipschitz. 

As a consequence, the optimal map $T$ transporting the Gaussian measure onto the law of $X$ satisfies $T'(x) \leq \frac{c}{\alpha}(1 + |x|^2)$. 
\end{theorem}

This result can be extended to the non-centered case just by translating. 

Note that a measure that is the image of the exponential measure by a Lipschitz map necessarily satisfies Cheeger's inequality, so the first part of this statement is actually an equivalence.

\begin{proof}
Showing an upper bound on the derivative of the map is the same as proving a lower bound on the derivative of the inverse map $\tilde{T}$ (which is the optimal map sending $\mu$ onto the exponential measure). A computation shows that, if we denote by $f$ the density of the law of $X$, we have for $x$ positive
$$\tilde{T}'(x) = \frac{2f(x)}{1 - F_{\mu}(x)}.$$
Since $\mu$ has median $0$, the Cheeger inequality implies that for $x \geq 0$ we get $f(x) \geq \frac{\alpha}{2}(1-F_{\mu}(x))$ and the first part of the result immediately follows, after applying the same reasoning for negative $x$.

The second part can be deduced just by using the fact that in dimension 1 the optimal map from the Gaussian onto $\mu$ is the composition of the map from the Gaussian onto the exponential measure with $T$. 
\end{proof}

The same argument can be generalized to radially symmetric random variables with log-concave law: 

\begin{proposition}
Assume that $X$ is a radially symmetric random variable, whose law  is log concave and satisfies a Cheeger inequality with constant $\alpha$. There exists a constant $c$ (independent of $X$) such that the optimal map $T$ transporting the Gaussian measure onto $X$ satisfies $\lambda_{\mathrm{max}}(\nabla T(x)) \leq c\alpha^{-1}\sqrt{1 + |x|^2}$. 
\end{proposition}

\begin{remark} \label{remark_bobkov}
Bobkov \cite{bobkov2003radial} showed that the optimal constant $\lambda$ in the Poincar\'e inequality for a radially symmetric log-concave random variable satisfies $n\mathbb{E}[|X|^2]^{-1}/12 \leq \lambda \leq n\mathbb{E}[|X|^2]^{-1}$. Since for log-concave measures the square of the Cheeger constant and the Poincar\'e constant are equivalent \cite{ledoux2004buser}, the constant $\alpha$ in the Proposition always exists and is comparable to $\sqrt{n}\mathbb{E}[|X|^2]^{-1/2}$, up to universal constants. 
\end{remark}

\begin{proof}
Let $\mu_{\mathrm{rad}}$ be the law of $|X|$, and $\tilde{T}$ be the optimal transport sending $\gamma_{\mathrm{rad}} := \frac{r^{n-1}}{\sqrt{2\pi}}e^{-r^2/2}dr$ onto $\mu_{\mathrm{rad}}$. The Brenier optimal map sending the Gaussian onto the law of $X$ is then given by $x \longrightarrow \tilde{T}(|x|)x/|x|$. This can be checked by verifying that it sends the Gaussian measure onto the law of $X$ (which is a simple change of variable argument) and that it is the gradient of the convex function $H(|x|)$ with $H' = \tilde{T}$. Since the Brenier map is the only transport map that arises as the gradient of a convex function, $T$ is necessarily the Brenier map. 

We then compute the gradient of the map $T$, which is given by $\nabla T(x)v = \left(\frac{v}{|x|} - \frac{x}{|x|^3}\langle v, x \rangle\right)\tilde{T}(|x|) + \frac{x}{|x|^2}\tilde{T}'(|x|)\langle x, v\rangle$. Since $T(0) = 0$, using the mean value theorem, we therefore have $\nabla T(x)v = \left(v - \frac{x}{|x|^2}\langle v, x \rangle \right)\tilde{T}'(t) + \frac{x}{|x|^2}\tilde{T}'(|x|)\langle x, v \rangle$ for some $t \in (0, |x|)$. From this, we see that to prove the desired upper bound on the eigenvalues of $\nabla T$, it is enough to show that $\tilde{T}'(r) \leq c\sqrt{1 + r^2}$.  

To prove this bound, we consider the symmetrized versions of $\mu_{\mathrm{rad}}$ and $\gamma_{\mathrm{rad}}$ by extending them by symmetry to $\R$, and dividing the density by $2$ so that they are still probability measures. We denote these measures by $\hat{\mu}_{rad}$ and $\hat{\gamma}_{rad}$. These measures are still log-concave, and their Cheeger constants are comparable to those of the original measures, up to a universal constant, via Bobkov's theorem we mentioned in Remark \ref{remark_bobkov}. Moreover, their median is located at $0$. We also extend $\tilde{T}$ to $\R$ by antisymmetry, and denote the function we obtain by $\hat{T}$. It is easy to check that $\hat{T}$ is the optimal map sending $\hat{\gamma}_{\mathrm{rad}}$ onto $\hat{\mu}_{\mathrm{rad}}$. 

Following the same arguments as for the one-dimensional case, denoting by $p$ the density of $\hat{\mu}$, we have for $r \geq 0$
\begin{align*}
\hat{T}'(r) &\geq Cr^{n-1}e^{-r^2/2}/p(F_{\hat{\mu}_{rad}}^{-1}(F_{\hat{\gamma}_{rad}}(r)) \\
&\leq C\alpha r^{n-1}e^{-r^2/2}F_{\hat{\gamma}_{rad}}(r)^{-1}(1 - F_{\hat{\gamma}_{rad}}(r))^{-1}\\
&\leq C\alpha r^{n-1}e^{-r^2/2}(1 - F_{\gamma_{rad}}(r))^{-1}\\
&\leq C\alpha \sqrt{1 + r^2}
\end{align*}
where we have used the estimate $1 - F_{\gamma_{rad}}(r) \geq Ce^{-r^2/2}r^{n-2}$ for large enough $r$, and $C$ was some positive constant that changed from line to line.
\end{proof}

\begin{remark}
In this proof, the log-concavity is only used to ensure that $\hat{\mu}_{rad}$ satisfies a Cheeger inequality with a constant comparable to $\alpha$. This is not necessarily the case for non log-concave radially symmetric measures (for example, the uniform measure on a two-dimensional annulus). 
\end{remark}
 
 \section*{Acknowledgment}
 T.~C.~and A.~P.~were supported  in part by NSF grants CCF-1528132 and CCF-0939370 (Center for Science of Information).  M.~F.~was supported by NSF FRG grant DMS-1361122.

\section{An estimate for the log-determinant function}

\begin{definition} \label{def:strconv}
A twice differentiable function $f: \dom f \to \mathbb{R}$ is  said to be ${\sf m}(x, y)$-strongly convex between $x, y \in \dom f$ if $\nabla^2 f( t x + (1 - t) y ) \geq {\sf m}(x, y) \mathrm{I}$, for all $t \in [0,1]$.
\end{definition}

\begin{lemma} \label{lem:strconv}
For all $t \in [0,1]$, a  ${\sf m}(x, y)$-strongly convex function $f$ between $x$ and $y$ satisfies
\begin{align}
t f(x) + (1 - t) f(y) \geq f(t x + (1 - t) y) + t (1 - t) \frac{{\sf m}(x, y)}{2} | x - y |^2.
\end{align}
\end{lemma}
\begin{proof}
The Taylor series expansion of $f$ for any two points $a, b \in \dom f$ yields
\begin{align}
f(a) &= f(b) + \langle\nabla f(b), b - a \rangle + \frac{1}{2} \langle b - a, \nabla^2 f(t_0a + (1- t_0)b)  (b - a) \rangle \label{eq1}\\
&\geq f(b) + \langle\nabla f(b), b - a \rangle + \frac{{\sf m}(a, b)}{2} | a - b |^2, \label{eq2}
\end{align}
where equation \eqref{eq1} holds for some $t_0 \in [0,1]$, and inequality \eqref{eq2} follows from Definition \ref{def:strconv}.

Denote $w \triangleq t x + (1 - t) y$. Applying inequality \eqref{eq2} twice yields
\begin{subequations}
\begin{align}
f(x) &\geq f(w) + \langle\nabla f(w), w - x \rangle + \frac{{\sf m}(x, w)}{2} | w - x |^2, \text{ and} \label{sub1}\\
f(y) &\geq f(w) + \langle\nabla f(w), w - y \rangle + \frac{{\sf m}(y, w)}{2} | w - y |^2. \label{sub2}
\end{align}
\end{subequations}
We now multiply equation \eqref{sub1} by $t$ and \eqref{sub2} by $(1 - t)$ and add them to obtain
\begin{align}
t f(x) + (1 - t) f(y) &\geq f(w) + \frac{t (1 - t)^2 {\sf m}(x, w) + t^2 (1 - t) {\sf m}(y, w)}{2} | y - x |^2.
\end{align}
By definition, we have ${\sf m}(x, w) \geq {\sf m}(x, y)$ and ${\sf m}(y, w) \geq {\sf m}(x, y)$, and this proves the lemma.
\end{proof}

\begin{proof}[Proof of Lemma \ref{lem:FrobNormIneq}]
The function  $f(\cdot) = -\log \det (\cdot)$ is known to be strictly convex and twice differentiable in the interior of the positive semidefinite cone. Substituting the definition of function $f$ into Lemma \ref{lem:strconv} yields
\begin{align}
\log \det(t A + (1 - t) B) \geq t \log \det(A)+ (1 - t) \log \det (B) + t (1 - t) \frac{{\sf m}(A, B)}{2} \| A - B \|_F^2.
\end{align}
It is a standard fact that $\nabla^2 f(M) = M^{-1} \otimes M^{-1}$, with $\otimes$ denoting the Kronecker product. The minimum eigenvalue of this Kronecker product is given by $1/\lambda^2_{\mathrm{max}}(M)$, where $\lambda_{\mathrm{max}}(M)$ is the largest eigenvalue of $M$. We therefore have $${\sf m}(A, B) \geq \min_{t \in [0,1]} \frac{1}{ \lambda^2_{\mathrm{max}}(t A + (1 - t)B)}.$$ Using the convexity of the maximum eigenvalue then yields the desired result.

\end{proof}

\bibliographystyle{IEEEtran}
\bibliography{stabilityBib}

% Generated by IEEEtran.bst, version: 1.13 (2008/09/30)
\begin{thebibliography}{10}
\providecommand{\url}[1]{#1}
\csname url@samestyle\endcsname
\providecommand{\newblock}{\relax}
\providecommand{\bibinfo}[2]{#2}
\providecommand{\BIBentrySTDinterwordspacing}{\spaceskip=0pt\relax}
\providecommand{\BIBentryALTinterwordstretchfactor}{4}
\providecommand{\BIBentryALTinterwordspacing}{\spaceskip=\fontdimen2\font plus
\BIBentryALTinterwordstretchfactor\fontdimen3\font minus
  \fontdimen4\font\relax}
\providecommand{\BIBforeignlanguage}[2]{{%
\expandafter\ifx\csname l@#1\endcsname\relax
\typeout{** WARNING: IEEEtran.bst: No hyphenation pattern has been}%
\typeout{** loaded for the language `#1'. Using the pattern for}%
\typeout{** the default language instead.}%
\else
\language=\csname l@#1\endcsname
\fi
#2}}
\providecommand{\BIBdecl}{\relax}
\BIBdecl

\bibitem{shannon48}
C.~E. Shannon, ``A mathematical theory of communication,'' \emph{Bell Syst.
  Tech. J.}, vol.~27, pp. 623--656, 1948.

\bibitem{stam1959some}
A.~Stam, ``Some inequalities satisfied by the quantities of information of
  {F}isher and {S}hannon,'' \emph{Information and Control}, vol.~2, no.~2, pp.
  101--112, 1959.

\bibitem{blachman1965convolution}
N.~M. Blachman, ``The convolution inequality for entropy powers,''
  \emph{Information Theory, IEEE Transactions on}, vol.~11, no.~2, pp.
  267--271, 1965.

\bibitem{lieb1978proof}
E.~H. Lieb, ``Proof of an entropy conjecture of {W}ehrl,'' \emph{Communications
  in Mathematical Physics}, vol.~62, no.~1, pp. 35--41, 1978.

\bibitem{carlen1991entropy}
E.~A. Carlen and A.~Soffer, ``Entropy production by block variable summation
  and central limit theorems,'' \emph{Communications in mathematical physics},
  vol. 140, no.~2, pp. 339--371, 1991.

\bibitem{saumard2014log}
A.~Saumard and J.~A. Wellner, ``Log-concavity and strong log-concavity: a
  review,'' \emph{Statistics surveys}, vol.~8, p.~45, 2014.

\bibitem{fathi2014quantitative}
M.~Fathi, E.~Indrei, and M.~Ledoux, ``Quantitative logarithmic {S}obolev
  inequalities and stability estimates,'' \emph{to appear in Disc. Cont. Dyn.
  Syst. Series A}, 2016.

\bibitem{cordero2015quantitative}
D.~Cordero-Erausquin, ``Transport inequalities for log-concave measures,
  quantitative forms and applications,'' \emph{to appear in Canad. J. Math.},
  2015.

\bibitem{toscani2015strengthened}
G.~Toscani, ``A strengthened entropy power inequality for log-concave
  densities,'' \emph{IEEE Transactions on Information Theory}, vol.~61, no.~12,
  pp. 6550--6559, 2015.

\bibitem{ball2012entropy}
K.~Ball and V.~H. Nguyen, ``Entropy jumps for isotropic log-concave random
  vectors and spectral gap,'' \emph{Studia Mathematica}, vol. 213, no.~1, pp.
  81--96, 2012.

\bibitem{ball2003entropy}
K.~Ball, F.~Barthe, and A.~Naor, ``Entropy jumps in the presence of a spectral
  gap,'' \emph{Duke Mathematical Journal}, vol. 119, no.~1, pp. 41--63, 2003.

\bibitem{courtadeStrengtheningISIT2016}
T.~A. Courtade, ``Strengthening the entropy power inequality,'' in \emph{2016
  IEEE International Symposium on Information Theory (ISIT)}, July 2016, pp.
  2294--2298.

\bibitem{courtade2016links}
------, ``Links between the logarithmic {S}obolev inequality and the
  convolution inequalities for entropy and fisher information,'' \emph{arXiv
  preprint arXiv:1608.05431}, 2016.

\bibitem{gross1975logarithmic}
L.~Gross, ``Logarithmic {S}obolev inequalities,'' \emph{American Journal of
  Mathematics}, vol.~97, no.~4, pp. 1061--1083, 1975.

\bibitem{carlen1991superadditivity}
E.~A. Carlen, ``Superadditivity of {F}isher's information and logarithmic
  {S}obolev inequalities,'' \emph{Journal of Functional Analysis}, vol. 101,
  no.~1, pp. 194--211, 1991.

\bibitem{indrei2013quantitative}
E.~Indrei and D.~Marcon, ``A quantitative log-{S}obolev inequality for a two
  parameter family of functions,'' \emph{International Mathematics Research
  Notices}, p. rnt138, 2013.

\bibitem{bobkov2014lsi}
S.~Bobkov, N.~Gozlan, C.~Roberto, and P.-M. Samson, ``Bounds on the deficit in
  the logarithmic sobolev inequality,'' \emph{Journal of Functional Analysis},
  vol. 267, pp. 4110--4138, 2014.

\bibitem{bolley2015lsi}
F.~Bolley, I.~Gentil, and A.~Guillin, ``Dimensional improvements of the
  logarithmic sobolev, talagrand and brascamp-lieb inequalities,''
  \emph{preprint}, 2015.

\bibitem{courtade2016entropy}
T.~A. Courtade, ``Entropy jumps for radially symmetric random vectors,''
  \emph{arXiv preprint arXiv:1608.05430}, 2016.

\bibitem{toscani1999sharp}
G.~Toscani and C.~Villani, ``Sharp entropy dissipation bounds and explicit rate
  of trend to equilibrium for the spatially homogeneous {B}oltzmann equation,''
  \emph{Communications in mathematical physics}, vol. 203, no.~3, pp. 667--706,
  1999.

\bibitem{villani2003cercignani}
C.~Villani, ``Cercignani's conjecture is sometimes true and always almost
  true,'' \emph{Communications in mathematical physics}, vol. 234, no.~3, pp.
  455--490, 2003.

\bibitem{godavarti2004convergence}
M.~Godavarti and A.~Hero, ``Convergence of differential entropies,'' \emph{IEEE
  Transactions on Information Theory}, vol.~50, no.~1, pp. 171--176, 2004.

\bibitem{bobylev1999rate}
A.~V. Bobylev and C.~Cercignani, ``On the rate of entropy production for the
  {B}oltzmann equation,'' \emph{Journal of statistical physics}, vol.~94, no.
  3-4, pp. 603--618, 1999.

\bibitem{Carlen201085}
E.~A. Carlen, M.~C. Carvalho, J.~L. Roux, M.~Loss, and C.~Villani, ``Entropy
  and chaos in the {K}ac model,'' \emph{Kinetic and Related Models}, vol.~3,
  no.~1, pp. 85--122, 2010.

\bibitem{villani2003topics}
C.~Villani, \emph{Topics in optimal transportation}.\hskip 1em plus 0.5em minus
  0.4em\relax American Mathematical Soc., 2003, no.~58.

\bibitem{villani2008optimal}
------, \emph{Optimal transport: old and new}.\hskip 1em plus 0.5em minus
  0.4em\relax Springer Science \& Business Media, 2008, vol. 338.

\bibitem{rioul2016yet}
O.~Rioul, ``Yet another proof of the entropy power inequality,'' \emph{arXiv
  preprint arXiv:1606.05969}, 2016.

\bibitem{knothe1957contributions}
H.~Knothe \emph{et~al.}, ``Contributions to the theory of convex bodies.''
  \emph{The Michigan Mathematical Journal}, vol.~4, no.~1, pp. 39--52, 1957.

\bibitem{rosenblatt1952remarks}
M.~Rosenblatt, ``Remarks on a multivariate transformation,'' \emph{The annals
  of mathematical statistics}, vol.~23, no.~3, pp. 470--472, 1952.

\bibitem{brenier1991polar}
Y.~Brenier, ``Polar factorization and monotone rearrangement of vector-valued
  functions,'' \emph{Communications on pure and applied mathematics}, vol.~44,
  no.~4, pp. 375--417, 1991.

\bibitem{mccann1995existence}
R.~J. McCann \emph{et~al.}, ``Existence and uniqueness of monotone
  measure-preserving maps,'' \emph{Duke Mathematical Journal}, vol.~80, no.~2,
  pp. 309--324, 1995.

\bibitem{figalli2010mass}
A.~Figalli, F.~Maggi, and A.~Pratelli, ``A mass transportation approach to
  quantitative isoperimetric inequalities,'' \emph{Inventiones mathematicae},
  vol. 182, no.~1, pp. 167--211, 2010.

\bibitem{caffarelli1990localization}
L.~A. Caffarelli, ``A localization property of viscosity solutions to the
  {M}onge-{A}mp\`ere equation and their strict convexity,'' \emph{Ann. of Math.
  (2)}, vol. 131, no.~1, pp. 129--134, 1990.

\bibitem{caffarelli1992regularity}
------, ``The regularity of mappings with a convex potential,'' \emph{Journal
  of the American Mathematical Society}, vol.~5, no.~1, pp. 99--104, 1992.

\bibitem{caffarelli2000monotonicity}
------, ``Monotonicity properties of optimal transportation and the {FKG} and
  related inequalities,'' \emph{Communications in Mathematical Physics}, vol.
  214, no.~3, pp. 547--563, 2000.

\bibitem{dowson1982frechet}
D.~Dowson and B.~Landau, ``The {F}r{\'e}chet distance between multivariate
  normal distributions,'' \emph{Journal of multivariate analysis}, vol.~12,
  no.~3, pp. 450--455, 1982.

\bibitem{colombo2016lipschitz}
M.~Colombo, A.~Figalli, and Y.~Jhaveri, ``Lipschitz changes of variables
  between perturbations of log-concave measures,'' \emph{preprint}, 2016.

\bibitem{kolesnikov2014hessian}
A.~V. Kolesnikov, ``Hessian metrics, {$CD(K,N)$}-spaces, and optimal
  transportation of log-concave measures,'' \emph{Discrete Contin. Dyn. Syst.},
  vol.~34, no.~4, pp. 1511--1532, 2014.

\bibitem{bonnefont2016sphere}
M.~Bonnefont, A.~Joulin, and Y.~Ma, ``Spectral gap for spherically symmetric
  log-concave probability measures, and beyond,'' \emph{J. Funct. Anal.}, vol.
  270, no.~7, pp. 2456--2482, 2016.

\bibitem{bonnefont20161d}
------, ``A note on spectral gap and weighted poincaré inequalities for some
  one-dimensional diffusions,'' \emph{ESAIM: Probability \& Statistics},
  vol.~20, pp. 18--29, 2016.

\bibitem{buser1982isoperimetry}
P.~Buser, ``A note on the isoperimetric constant,'' \emph{Ann. Sci. \'Ecole
  Norm. Sup. (4)}, vol.~15, no.~2, pp. 213--230, 1982.

\bibitem{ledoux2004buser}
M.~Ledoux, ``Spectral gap, logarithmic sobolev constant, and geometric
  bounds.'' \emph{Surveys in differential geometry}, vol.~9, pp. 219--240, 204.

\bibitem{milman2009}
E.~Milman, ``On the role of convexity in isoperimetry, spectral gap and
  concentration,'' \emph{Inventiones Mathematicae}, vol. 177, no.~1, pp. 1--43,
  2009.

\bibitem{bobkov2003radial}
S.~G. Bobkov, ``Spectral gap and concentration for some spherically symmetric
  probability measures,'' \emph{Geometric Aspects of Functional Analysis,
  Lecture Notes in Math.}, vol. 1807, pp. 37--43, 2003.

\end{thebibliography}

\end{document}